\documentclass[conference]{IEEEtran}
\usepackage{times}

\usepackage{amsmath}
\usepackage{amssymb}
\usepackage{booktabs}
\usepackage{algpseudocode}
\usepackage{algorithm}
\usepackage{enumerate}
\usepackage{subcaption}
\usepackage{gensymb}
\usepackage{amsthm}
\let\labelindent\relax
\usepackage{enumitem}
\usepackage{booktabs}
\usepackage{multirow}
\usepackage{adjustbox}
\usepackage{tabularx}
\usepackage{lipsum} 
\usepackage{pifont}
\usepackage[most]{tcolorbox}
\usepackage[font=small,labelfont=bf]{caption}

\usepackage[numbers]{natbib}
\usepackage{multicol}
\usepackage[bookmarks=true]{hyperref}

\newtheorem{definition}{Definition}[section]

\newtheorem{proposition}[definition]{Proposition}

\newtheorem{lemma}[definition]{Lemma}

\newtheorem{theorem}[definition]{Theorem}

\begin{document}

\title{Fast Uniform Dispersion of a Crash-prone Swarm}


\author{\authorblockN{Michael Amir}
\authorblockA{Technion - Israel Institute of Technology\\
Email: ammicha3@cs.technion.ac.il}
\and
\authorblockN{Alfred M. Bruckstein}
\authorblockA{Technion - Israel Institute of Technology\\
Email: freddy@cs.technion.ac.il}}


%

\maketitle

\begin{abstract}
We consider the problem of completely covering an unknown discrete environment with a swarm of asynchronous, frequently-crashing autonomous mobile robots. We represent the environment by a discrete graph, and task the robots with occupying every vertex and with constructing an implicit distributed spanning tree of the graph. The robotic agents activate independently at random exponential waiting times of mean $1$ and enter the graph environment over time from a source location. They grow the environment's  coverage by `settling' at empty locations and aiding other robots' navigation from these locations. The robots are identical and make decisions driven by the same simple and local rule of behaviour. The local rule is based only on the presence of neighbouring robots, and on whether a settled robot points to the current location. Whenever a robot moves, it may crash and disappear from the environment. Each vertex in the environment has limited physical space, so robots frequently obstruct each other. 

Our goal is to show that even under conditions of asynchronicity, frequent crashing, and limited physical space, the simple mobile robots complete their mission in linear time asymptotically almost surely, and time to completion degrades gracefully with the frequency of the crashes. Our model and analysis are based on the well-studied ``totally asymmetric simple exclusion process'' in statistical mechanics. 
\end{abstract}

\IEEEpeerreviewmaketitle\captionsetup{belowskip=3pt,aboveskip=1pt}

\section{Introduction}

In swarm robotics, a vast number of autonomous mobile robots cooperate to achieve complex goals \cite{hsiang}. Individual members of the swarm are usually assumed to be simple, expendable, and computationally limited, and to move and act according to online local rules of behaviour. In this work, our goal is to formally study the ability of a simple, ``two-layered'' swarm-robotic system to complete an environmental coverage task called \textit{uniform dispersal} assuming  asynchronicity and that robots may crash whenever they attempt to move.

``Coverage'' algorithms that enable a single- or multi-robot system to cover or explore unknown or dynamically uncertain environments are an important topic in mobile robotics. There has been great interest in applications to, for example, mapping \cite{amigoni2010information, howard2002, corah2017efficient}, servicing and surveillance \cite{chen2013temporal}, or search and rescue operations \cite{jorgensen2017risk, calisi2005autonomous,basilico2011exploration, basilico2016semantically}, and a rich body of theoretical work exists (we refer the reader to the surveys \cite{altshuler2018introduction, galceran2013survey}). A natural coverage problem for robotic swarms is  the  \textit{uniform dispersal} problem, introduced in \cite{hsiang}. In \textit{uniform dispersal}, many robotic agents enter an unknown discrete graph environment over time via one or several source locations and are tasked to eventually occupy every vertex of the graph with a robot while avoiding collisions. 

Swarms are often claimed to be highly fault-tolerant, as redundancy and sheer numbers  can enable the swarm to go on with its mission even if many  robots malfunction \cite{winfield2006faulttolerant}. However, as the size of a robotic fleet grows, so too does the opportunity for error. Specifically, three different complications that arise in multi-robot systems are further exacerbated in the swarm setting:

\textbf{Asynchronicity.} As the number of robots grows, coordinating the robots'  actions becomes a formidable task, as their actions and internal clocks can become highly unsynchronized. 

\textbf{Crashes.} We cannot expect to release a huge swarm of simple robots to an unknown environment without the occurrence of hardware or software faults that may cause robots to crash. 

\textbf{Traffic.} To avoid collisions, we do not wish for there to be too many robots crowding a given area, and so  mobile robots should maintain safe distances from each other. In restricted physical environments, such requirements cause traffic delays, as robots must wait for other robots to move away before entering a target location.

Such challenges are discussed as a central direction of research for swarm robotics in \cite{peleg2005distributed}. If the number of errors scales with the number of robots, are swarms ``worth the trouble''? The purpose of this work is to give a perspective on this question via a formal mathematical analysis. We study, in an abstract setting, the ability of a simple local rule to achieve \textit{uniform dispersal} in the presence of crashes and asynchronicity. We are specifically interested in how the frequency of crashes affects the time to mission completion.

We first describe a ``two-layered'' rule of behaviour for swarms that is capable of achieving uniform dispersal. Using this algorithm, we show that a swarm can complete its mission quickly and reliably in unknown discrete environments, even in the presence of asynchronicity and frequent crashes. Hence, we claim that in our setting, many robots can win against many errors. In the spirit of swarm robotics, the algorithm relies only on local information to dictate robots' actions and is  quite simple. This simplicity makes it amenable to analysis.

Our swarm consists of a large reservoir of simple, anonymous, identical, and autonomous mobile robots that enter the environment over time via a source location $s$. The robots move across a discrete  environment represented by an \textit{a priori unknown} connected graph $G$ whose vertices represent spatial locations. The robots  gradually expand their coverage of the environment by occupying certain locations and assisting other nearby robots in navigational tasks using a local, indirect communication scheme.

The swarm's robots switch between two modes: mobile and settled. The settled robots act as `nodes' of the current coverage of the graph environment, and the mobile robots move between locations with settled robots until they can find a new location where they themselves can settle. The settled `robot nodes' are capable of \textit{pointing to} (``\textit{marking}'') a single neighbouring location where there is another settled robot. ``Marking'' is understood to be a generic capability of the robots and could be accomplished by many different technologies, such as local radio communication or visual sensing; we refer to the Related Work section for possible implementations. 

As more and more mobile robots become settled, their marks serve as a navigational network of the environment that is utilised by the remaining mobile robots. The mobile robots are capable of sensing the number of robots in neighbour locations, and sensing when a settled robot is pointing to (marking) their location. They rely only on this information to make decisions. Hence, they operate in a GPS-denied, low memory setting, meaning they act based only on local communications and  local geographic features. The robots are tasked with settling at every vertex of $G$,  and constructing an implicit \textit{spanning tree} of $G$ via the settled robots and their pointer marks. 

There are no restrictions on $G$ as long as it is connected. In principle, different robots need not even agree on the graph representation of their environment for our algorithm to work (e.g., in case they gradually build it from local sensory data), as the settled robots gradually construct a spanning tree which all robots agree on and use to move between locations. We assume, for simplicity, that they share the same representation. 

\textbf{Physical constraints and asynchronicity.} We model the mobile robots as activating  repeatedly at stochastic, independent exponential waiting times of rate $1$. When a robot activates, it may move or move-and-settle at a nearby location (once  a robot settles, it remains stationary). We assume the physical constraint that any given location may contain no more than a single mobile robot and a single settled robot (and perhaps a number of crashed robots). Frequent traffic obstructions occur as robots block each other off from progressing.

This model of asynchronicity and limited vertex capacity in a graph environment is motivated by the \textit{totally asymmetric simple exclusion process} (TASEP) in statistical mechanics. There is an extensive literature on this process as a model for a great variety of transport phenomena, such as traffic flow \cite{chowdhury2000statistical} and biological transport \cite{chou2011biologicaltasep}. Rigorous exact and asymptotic results for TASEP are known \cite{Johansson2000,  tracy2009asymptotics}, and our analysis technique shall be to compare our swarm's performance to a two-layered TASEP-like process. Since our robots are mostly in a state of "traffic flow" (waiting for other robots to move), references such as \cite{chowdhury2000statistical} suggest that our model in fact  captures many of the relevant traffic phenomena that will occur in real life implementations.

\textbf{Adversarial crashing.} Similar to, e.g., \cite{jorgensen2017risk}, we consider a risky traversal model where robots may crash whenever they try to move across an edge. We assume robots remain safe when not moving, as remaining put is less risky than travelling (in fact, we need just the weaker assumption that settled robots, which \textit{never} move, are safe). To facilitate analysis, we assume crashed robots do not prevent travel between the graph's vertex hubs. This assumption is applicable when such robots can be manoeuvred around or pushed aside, or the crash causes the robot to disappear. For example, we may consider crashed air-based robots falling to the ground during exploration of an environment. Alternatively, in a ground robot scenario, we can with foresight expand vertex sizes to be big enough such that vertices can contain a small number of crashed robots in addition to the two active robots (and such local crashed robots are then bypassed using, e.g., local collision avoidance). 

Since, in our model, at most one new robot may arrive in the environment per time step, we assume that the number of crashes that occur is bounded by the current time $t$, and a parameter $c$ which reflects the frequency at which crashes occur over time. When $c$ is close to $1$, the \textit{vast majority} of robots that enter the environment will crash before achieving anything.

Besides these limitations, we assume nothing more about the  crashes that occur. In particular, a \textit{virtual adversary} may choose crashes so as to be as obstructing as possible.

\textbf{Results.} We describe a local rule of behaviour (Algorithm \ref{alg:localrule}) that can achieve uniform dispersion, even in the presence frequent crashes and traffic obstructions. The rule is easy to understand and implement and is well-suited for a swarm of simple robots, mimicking a kind of branching depth-first search. In many mobile robot systems one wishes to construct a spanning tree of the environment for purposes of mapping, routing or broadcasting \cite{abbas2006distributedspanning, agmon2006spanning, broder1989spanning, mapdrawing, gabriely2001spanning}. Our  rule achieves this as well, by having robots act as nodes of the tree, and making them aware of their immediate descendants. Our goal is to study how crashes, asynchronicity, and traffic affect the swarm's performance under this rule of behaviour. 

We prove that our robots are able to complete their mission in time linear in the size of the environment, and that performance degrades gracefully (by a factor $(1-c)^{-1}$) with frequency of crashes. Given our assumptions and algorithm, it is not surprising that the robots can complete the dispersion assuming some crashes; rather, we show that even with many frequent crashes, the robots can still do so efficiently. 

Specifically, let $n$ be the number of vertices in the environment $G$. We prove that dispersal completes before time $8 \cdot \big((1-c)^{-1} + o(1)\big)n$ asymptotically almost surely (meaning with probability approaching $1$ as $n$ grows)--a worst-case bound on performance. No dispersal algorithm can complete in less than $O(n)$ expected time, since this is the time it takes to even explore $n$ vertices, so when there are no crashes (but still there is traffic and asynchronicity) this bound is asymptotically tight. For, say, $c = 0.5$, we expect up to (roughly) $50$\% of robots to crash before achieving anything, and our analysis says that therefore the swarm will take twice as long to achieve dispersal. This seems intuitive, but consider that the robots that eventually crash are (uselessly) present in the environment in the time leading to the crash, blocking other robots from entering or progressing. The analysis says that nevertheless, the ability of the rest of the swarm to achieve its goal is not disproportionately worsened.

To the best of our knowledge, with or without crashes, we are the first to consider a non-synchronous setting for the uniform dispersal problem where time to completion can explicitly be bounded, hence also the first to give explicit performance guarantees in a non-synchronous setting. In an asynchronous as opposed to a synchronous setting, there are many more possible configurations that the robots might exist in, which makes the analysis more difficult. We believe the references and techniques from statistics \cite{Johansson2000,  tracy2009asymptotics, chowdhury2000statistical} might be of general interest for tackling these kinds of topics.

Our analysis extends also to a synchronous time setting, and to the case where robots enter the environment from \textit{multiple} locations. Multiple entrance locations result instead in the robots constructing instead an implicit \textit{spanning forest}. In both these settings, dispersion completes faster. The bound on performance  we derive for the synchronous case is exact.

Finally, we confirm our findings by numerically simulating our system in a number of environments and measuring performance.  

\vspace{-0.5em}
\subsection{Related work}
\vspace{-0.25em}
Uniform dispersal was introduced by Hsiang et al. in \cite{hsiang} for discrete grid environments of connected pixels (but their work can be extended to arbitrary graph environments). They considered a synchronous time setting where robots are allowed to send short messages to nearby robots, and showed time-optimal algorithms for this setting. Many variations have since been studied. Barrameda et al. extend the problem to the asynchronous setting with no explicit non-visual communication  \cite{barrameda2013uniform, barraswarm1}. Recent works include dispersal with weakened sensing \cite{hideg2017uniformtime}, dispersal in arbitrary graph environments \cite{dispersalgraphs2019}, and dispersal under energy constraints  \cite{arxivminimizingtravel}. Our model differs from previous work on several central points, including the presence of crashes, the two layers, and the ability to mark neighbours. Marking is weaker than the radio  communication available to robots in \cite{hsiang}, that enables robots to transfer many bits of data locally, but stronger than the indirect, visual communication assumed in several other works. 

\begin{table}[]
\centering
\fontsize{7}{7}\selectfont
\begin{tabularx}{\textwidth}{p{1cm}p{1.1cm}p{0.7cm}p{1.4cm}p{1cm}p{0.8cm}}
   Reference                                       & Environment                       & Time                  & Communication              & Makespan         & Crashes               \\
\centering\arraybackslash\cite{hsiang}               & Arbitrary          & Synch.             & Radio           & $O(n)$           & \centering\arraybackslash x                   \\
\centering\arraybackslash\cite{arxivminimizingtravel} & Hole-less grid & Synch.              & Visual                   & $O(n)$           & \centering\arraybackslash x                    \\
\centering\arraybackslash\cite{hideg2017uniformtime}  & Grid       & Synch.              & Visual                     & $O(n)$           & \centering\arraybackslash x                    \\
\centering\arraybackslash\cite{barraswarm1}           & Hole-less grid & Asynch.             & Visual                     &  undefined        & \centering\arraybackslash x                    \\
\centering\arraybackslash\cite{barrameda2013uniform}  & Grid         & Asynch.             & Radio & undefined         & \centering\arraybackslash x                 \\
Our work                                      & Arbitrary          & Stochastic Asynch. & Marking  & $O(n)$ & \centering\arraybackslash\checkmark
\end{tabularx}
\newline
\newline
\vspace{-1em}
\caption{A comparison of works on uniform dispersal.}
\vspace{-3em}
\label{tablecompare}
\end{table}

Because of differences in the settings, assumptions, and constraints, quantitative comparison of works on uniform dispersal is very difficult. Table \ref{tablecompare} gives a \textit{rough, non-exhaustive} overview of some differences, such as the supported kinds of environments (grid environment, hole-less grid environment, or arbitrary graph environment), synchronous versus asynchronous time, expected makespan (i.e., how long it takes the robots to complete their mission), and whether crashes are considered in the model.

Robotic coverage, patrolling, and exploration with adversarial interference, as well as crashes, have been studied in different problem settings from our own. Agmon and Peleg studied a gathering problem for robots where a single robot may crash \cite{agmon2006fault}, and gathering with multiple crashes was later discussed by Zohir et al. in a similar setting \cite{bouzid2013gathering}. Robotic exploration in an environment containing threats has been studied in \cite{yehoshua2013robotic, yehoshua2015frontier}. Moreover, adversarial crashes of processes are often studied in general distributed algorithms (e.g., \cite{delporte2011disagreement}). Differing from many of these works, we study a situation where the number of crashes scales with the mission's complexity (the time it takes to cover the environment), and where even the vast majority of robots may crash. However, to enable this, we assume access to a huge reservoir of robots waiting to replace crashed robots--i.e., a robotic swarm.
 
Robotic coverage in various hazardous or adversarial GPS-denied settings has become an important topic in recent decades, since this opens the possibility of deploying robotic swarms in the real world, outside laboratory conditions \cite{galceran2013survey, altshuler2018introduction, agmon2017adversarialrobotic}. Theoretical and empirical results about the performance of swarms in such settings may help inform our expectations of real world swarm-robotic fleets. To implement such systems in practice, the robots themselves must be capable of relative visual localization. This poses a technical challenge, as considerations of depth, angle of view, and persistent coverage come into play. In \cite{biswas2012depth} a system of relative visual localization for mobile ground vehicles with low computing power is proposed. The system enables autonomous ground vehicles to navigate their environment while avoiding obstacles. In \cite{saska2017system} a relative visual localization technique is developed for small quadcopters, with similar capabilities. In \cite{prorok2011reciprocal} the authors discuss a localization algorithm for lightweight asynchronous  multi-robot systems with lossy communication. These are examples of the techniques that may be used for the sensors of robots in such systems as the one described in this paper (see similar discussion in \cite{arxivminimizingtravel}). 
 
A fascinating introduction to TASEP-like processes and their connection to other fields is \cite{kriecherbauer2010pedestrian}.

\section{Model and System}
\label{modelsection}

We consider a swarm of mobile robotic agents performing world-embedded calculations on an unknown discrete environment represented by a connected graph $G$. The vertices of $G$ represent spatial locations, and the edges represent connections between these locations, such that the existence of an edge $(u,v)$ indicates that a robot may move from $u$ to $v$. 

We assume an infinite collection of robots (also referred to as `agents') attempt to enter $G$ over time through a  \textit{source} vertex $s \in G$. The robots are identical and execute the same algorithm. They begin in the  \textit{mobile} state, and eventually enter the \textit{settled} state. Settled robots are stationary, and are capable of \textit{marking} a neighbouring vertex that contains another settled robot. Mobile robots move between the vertices of $G$ and sometimes crash while in motion. They are oblivious, and decide where to move based only on local information provided by their sensors: the number of robots at neighbouring vertices, and whether any of the neighbouring settled robots mark their current location. Each vertex has limited capacity: it can contain at most one settled and one mobile robot.

Mobile robots are only allowed to move to a neighbouring vertex when they are \textit{activated}. Each robot, including robots outside $G$, reactivates infinitely often and independent of other robots, at random exponential waiting times of mean $1$. 

When $s$ contains less than two robots, robots from outside $G$ attempt to enter it when they are activated. It is convenient to give the robots arbitrary labels $A_1, A_2, \ldots$ and assume that $A_i$ cannot enter $s$ before all robots with lower indices entered or crashed. This assumption makes the analysis simpler, but the performance bound we prove in this work holds also for the entrance model where robot entrance depends only on which robot is activated first. Hence, whenever the current lowest-index robot outside of $G$ activates and there is no \textit{mobile} robot at $s$, it moves to $s$. If $s$ is completely empty, the robot settles upon arrival and becomes the root of the spanning tree. Otherwise it remains a mobile robot. 

We denote by $G(t)$ the graph whose vertices are vertices of $G$ containing settled robots at time $t$, and there is a directed edge $(u,v) \in G(t)$ if $u$ is marked by a settled robot at $v$. The goal of the robots is to reach a time $T$ wherein $G(T)$ is a spanning tree of the entire environment $G$. The \textit{makespan} of an algorithm is the first time $T_0$ when this occurs. 

Crashes are modelled as follows: when a robot $A_i$ is activated and attempts to enter $s$ or move from $u$ to $v$ via the edge $(u,v)$, occasionally an \textit{adversarial event} occurs, causing the deletion of $A_i$ from $G$.  Robots do not crash unless attempting to move. Hence, mobile robots are volatile but settled robots are safe. This assumption is somewhat stronger than necessary: our results still hold if mobile (but not settled) robots are allowed to crash while they stay put, but this tediously lengthens the analysis. We assume the number of adversarial events before time $t$ is bounded by a fraction of $t$. Adversarial events may otherwise be as inconvenient as possible: we may assume there is an \textit{adversary} choosing crashes to maximize the makespan of our algorithm. 

Unless stated otherwise, when discussing the configuration of robots ``at time $t$'', we always refer to the configuration before any activation at time $t$ has occurred.

\section{Dispersal and Spanning Trees}

We study a simple local behaviour (Algorithm \ref{alg:localrule}) that disperses robots and incrementally constructs a distributed spanning tree of $G$. The rule determines the behaviour of \textit{mobile} robots whenever they are activated (settled robots merely remain in place and continue to mark their target). We prove that using this rule, the makespan is linear in the number of vertices of $G$ asymptotically almost surely, and that performance degrades gracefully with the density of crashes.

\begin{algorithm}[!htb]
  \caption{Local rule for a mobile robot $A$.}
  \begin{algorithmic}
    \State Let $v$ be the current location of $A$ in $G$ (if $A$ is outside $G$, see Section \ref{modelsection}).
    \If{a neighbour $u$ of $v$ contains exactly one robot, and this robot marks  $v$}
        \State Attempt to move to $u$. 
    \ElsIf{a neighbour $u$ of $v$ contains no robots}
        \State Attempt to move to $u$ and become \textit{settled} if no crash. 
        \State \textit{Mark} the vertex $v$.
    \Else
        \State Stay put.
    \EndIf
  \end{algorithmic}
  \label{alg:localrule}
\end{algorithm}

The rule grows $G(t)$ as a partial spanning tree of $G$. It acts as a kind of depth first search that splits into parallel processes whenever a mobile robot is blocked by another mobile  robot. Every vertex of the tree $G(t)$ is marked by settled robots at its descendants. Mobile robots follow these marks to discover the leaves of the current tree $G(t)$ and expand it. Robots grow the tree by settling at unexplored vertices that then become new leaves. Our main result is Theorem \ref{alg1performancethm}:

\begin{theorem}
If for all $t$ the number of adversarial events before time $t$ is allowed to be at most $ct/4$, $0 \leq c < 1$, then the makespan of Algorithm \ref{alg:localrule} over graph environments with $n$ vertices is at most $8((1-c)^{-1}+o(1))n$ asymptotically almost surely as $n \to \infty$. \label{alg1performancethm}
\end{theorem}

Figure \ref{fig:simulation} shows an execution of our algorithm on a grid environment with  $n=62$ square vertices (white region) and obstacles (blue region). We allowed a naive adversary to arbitrarily delete at most $ct/4$ robots before time $t$, with $c = 0.8$. This corresponded to a deletion of 56\% of robots that entered the environment before the makespan. In a more constrained topology (such as a path graph, see Section \ref{pathgraphsection}), the robots would progress more slowly, and a greater percentage would be deleted. The makespan (bottom right figure) was $613$, consistent with the upper bound of Theorem \ref{alg1performancethm}. After the  the spanning tree completes, robots keep entering the region until there are two robots at every vertex. This is related to the ``slow makespan'', which we will later define. The slow makespan was 831. See Section \ref{simulationsection} for more  simulations. 

\begin{figure}[htb]
    \centering 
\begin{subfigure}{0.13\textwidth}
  \includegraphics[width=\linewidth]{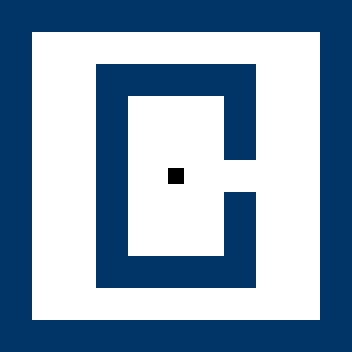}
  \label{fig:1}
\end{subfigure}\hfil 
\begin{subfigure}{0.13\textwidth}
  \includegraphics[width=\linewidth]{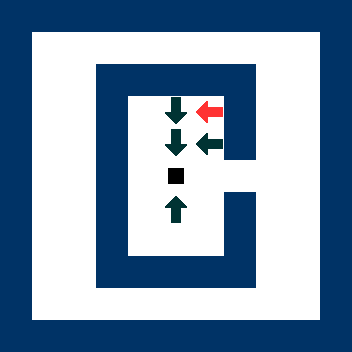}
  \label{fig:2}
\end{subfigure}\hfil 
\begin{subfigure}{0.13\textwidth}
  \includegraphics[width=\linewidth]{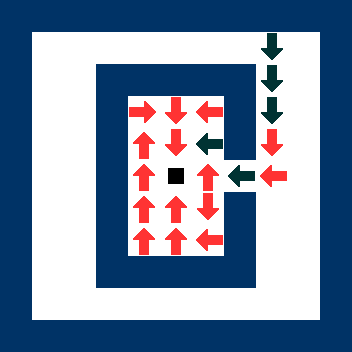}
  \label{fig:3}
\end{subfigure}

\medskip
\vspace{-1.5em}
\begin{subfigure}{0.13\textwidth}
  \includegraphics[width=\linewidth]{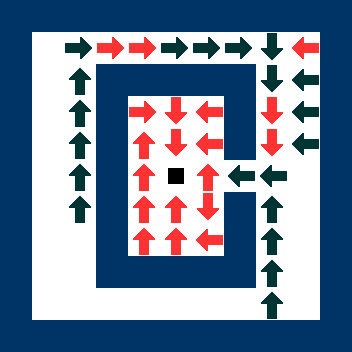}
  \label{fig:4}
\end{subfigure}\hfil 
\begin{subfigure}{0.13\textwidth}
  \includegraphics[width=\linewidth]{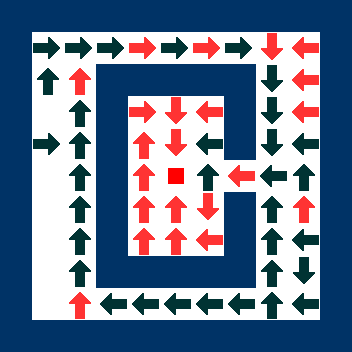}
  \label{fig:5}
\end{subfigure}\hfil 
\begin{subfigure}{0.13\textwidth}
  \includegraphics[width=\linewidth]{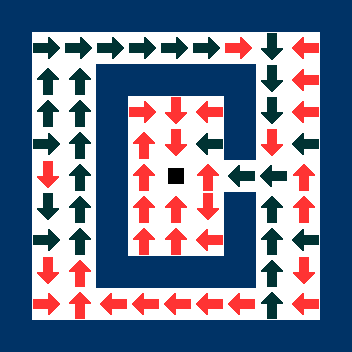}
  \label{fig:6}
\end{subfigure}
\vspace{-1em}
\caption{An execution of Algorithm \ref{alg:localrule} on a grid environment. The source is denoted by a square box in the center. The arrows denote settled robots, and their direction points to the adjacent location with a settled robot  that they mark. Red arrows indicate a mobile robot is on top of the settled robot (note that by the algorithm, a mobile robot will never occupy a vertex that does not have a settled robot). The environment is a priori unknown to the robots, and they construct a spanning tree representation over time.}
\label{fig:simulation}
\vspace{-2em}
\end{figure}

\subsection{Analysis}
\label{analysissection}

We study the makespan of Algorithm \ref{alg:localrule}. Some of the proofs are placed in the Appendix.

For the analysis, we will assume that robots from $A_1, A_2, \ldots$ that settle or crash keep being activated. This is a purely ``virtual'' activation: such robots of course do and affect nothing upon being activated.  We start with a structural Lemma:

\begin{lemma}
\label{treelemma}
$G(t)$ is a tree at all times $t$ with probability $1$.
\end{lemma}

\begin{proof} When the first robot enters and successfully settles, $G(t)$ contains only $s$. No settled robots are ever deleted, so $G(t)$ can only gain new vertices. Whenever a mobile robot settles, it extends the tree $G(t)$ by one vertex, connecting its current location $v$ to $G(t)$ via a single directed edge. By definition, the edge is directed from the vertex the settled robot \textit{marks}--which is its previous location--to $v$. This turns $v$ into a leaf of $G(t)$. With probability $1$ no two robots on $G$ activate at the exact same time, so no two robots settle the same vertex. Hence $G(t)$ remains a tree.
\end{proof}

\subsubsection{Event orders}
\label{eventordersection}
We explain how we intend to bound  the makespan. Our strategy shall be to use coupling to compare the performance of Algorithm \ref{alg:localrule} by the performance of different random processes of robots moving on different structures. Coupling is a technique in probability theory for comparing different random processes (see \cite{lindvall2002coupling}). 

The basic idea is this: whenever we run Algorithm \ref{alg:localrule} on $G$, we can log the exact times at which the robots activate, as well as the times adversarial events happen and which robots they affect. This gives us an \textit{order of events} $S$ sampled from some random distribution. Note that robots keep activating forever (but these activations do nothing once the graph is full), so $S$ is infinitely long.
  We then ``re-enact'' or \textbf{``simulate''} $S$ on a new environment (or several new environments)  involving the robots $A_1, A_2, \ldots$ by activating and deleting the robots according to $S$.
  
  To make things more precise, by ``simulating'' $S$ on different environments  we mean that we consider the \textit{coupled}  process $(G, G_2, \ldots, G_m)$ wherein different environments $G, G_2, \ldots G_m$ have robots that are \textit{paired} such that whenever $A_i$ in $G$ is scheduled for an activation or a deletion according to the event order $S$ ($S$ is simply an infinite list of scheduled activation and deletion times), the copies of $A_i$ in \textit{all} the environments $G_2, \ldots G_m$ also activate or are deleted. When the copies of $A_i$ are activated they act according to  Algorithm \ref{alg:localrule} with respect to their local neighborhood. Robots entrances are modelled as usual (Section \ref{modelsection}), but note that even if $A_i$ manages to enter $G$ following an activation, its copy  might not enter its own environment because in that environment the entrance is blocked, or there is a lower-index robot waiting to enter. During Algorithm \ref{alg:localrule}'s analysis, we will often be talking about a deterministic event order $S$ being simulated over different environments. The end-goal, however, is to say something about the event order $S$ when it is randomly sampled from the  execution of Algorithm \ref{alg:localrule} on $G$. 
  
  The event order $S$ must be a \textit{possible} set of events that occurred during an execution of our algorithm on the base graph environment $G$. This means, due to our model, that a robot $A_i$ in $G$ will never be scheduled for deletion except at times when it is activated and attempts to move. However, while simulating $S$ on the environments $G_2, \ldots G_m$, we must be allowed to break the rules of the model: we might delete robots even when they don't attempt to move, or while they are outside of the new graph environment. Whenever we say ``for any event order $S$'', we mean event orders $S$ \textbf{that could have happened over $G$}. 
  
  In $S$, define $t_0$ to be the first time $A_1$ activates, $t_1$ to be the first time \textit{after $t_0$} that either $A_1$ or $A_2$ activate, and $t_i$ to be the first time $t > t_{i-1}$ that any robot in the set $\{A_1, \ldots, A_{i+1}\}$ is activated. 
  
  \begin{definition}
  The times $t_0 < t_1 < t_2 < \ldots$ in $S$ are called the \textbf{meaningful event times} of $S$.  
  \end{definition}
  
  For meaningful event times to be well-defined there must be a \textit{minimal} time $t > t_{i-1}$ where one of the robots ${A_1, \ldots, A_i}$ activates. Because the activation times of the robots are independent exponential waiting times of mean $1$, this is true with probability $1$ for a randomly sampled $S$. Moreover, with probability $1$,  at any time $t_i$ there is precisely one robot $A$ of $A_1, A_2, \ldots, A_{i+1}$ scheduled for activation by $S$. Because both these things are true with probability $1$, we assume they are true for any event order  $S$ referred to at any point in this analysis. This does not affect our main result (Theorem \ref{alg1performancethm}), which is probabilistic.  
  
  Our end-goal is randomly sample $S$ from $G$ and simulate it on four increasingly ``slower''  environments: $\mathcal{P}(n)$, $\mathcal{P}(\infty)$, $\mathcal{P}^*(\infty)$, $B$, so that all environments ($G$ and these four) are coupled. \textit{Meaningful event times} are so called because, prior to the first  activation of $A_i$, any of the robots $A_{i+1}, A_{i+2}, \ldots$ cannot enter or move in any of these environments, and activating them causes nothing. Hence, at any time $t$ which is not a meaningful event time, the configuration of robots cannot change (no robots move and no robots are deleted in any of the environments $S$ is simulated on). 
  
  The possibility to create an event order $S$ is the only reason we labelled the robots and made the assumption about entrance orders in Section \ref{modelsection}. 
  
\begin{figure}[htb]
\centering 
\includegraphics[width=0.5\textwidth]{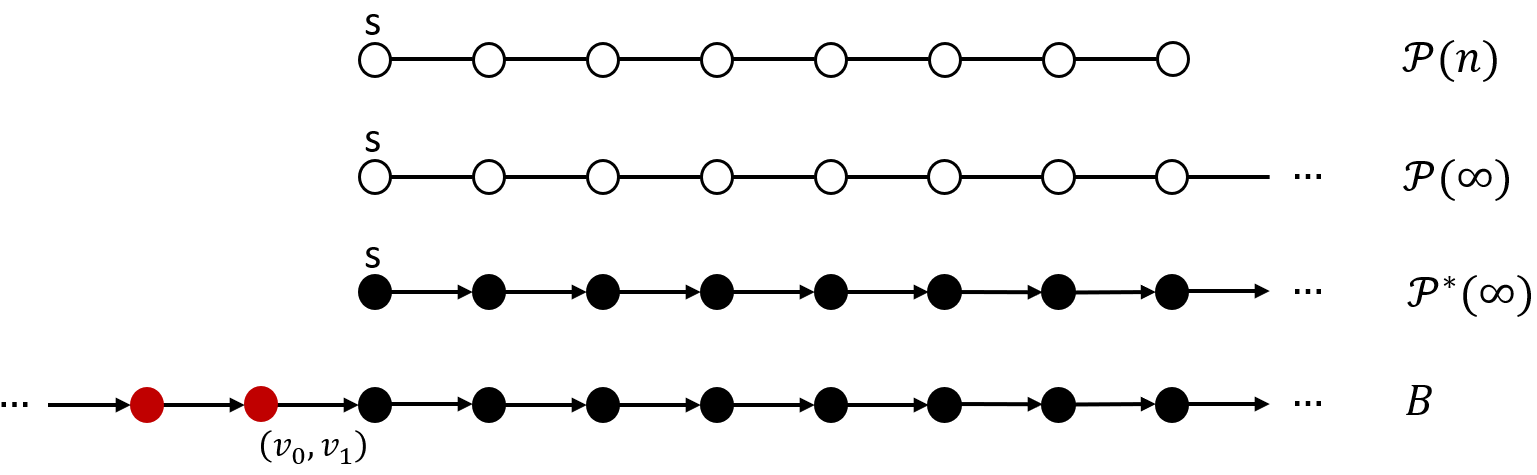}
\caption{The processes $\mathcal{P}(n), \mathcal{P}(\infty), \mathcal{P}^*(\infty), B$ that we will be interested in. White vertices are empty, black vertices contain a settled robot, and red vertices contain both a mobile and a settled robot. Edge directions indicate edge directions in $G(t)$. Note that $B$ does not have a source vertex.}
\label{fig:processes}
\vspace{-2em}
\end{figure}
  
\subsubsection{$\mathcal{P}(n)$ versus $G$}
\label{pathgraphvsGsection}

Let $n$ be the number of vertices of $G$. The path graph $\mathcal{P}(n)$ over $n$ vertices is a graph over the vertices $v_1 v_2 \ldots v_n$ such that there is an edge $(v_i, v_{i+1})$ for all $1 \leq i \leq n-1$. We simulate $S$ on the graph environment $\mathcal{P}(n)$ where the source vertex $s$ is $v_1$. Simulating $S$ on $\mathcal{P}(n)$ results in what is mostly a normal-looking execution of Algorithm \ref{alg:localrule} on $\mathcal{P}(n)$, but as discussed, it might lead to some oddities such as robots being deleted while they are still outside the graph environment.

Let us introduce some notation. $A_i^G$ refers to the copy of $A_i$ being simulated by $S$ on $G$, and $A_i^{\mathcal{P}(n)}$ is similarly defined. 

\begin{definition}
The \textit{depth} of $A_i^G$ at time $t$, written $d(A_i^G, t)$, is the number of times $A_i^G$ has \textit{successfully moved} before time $t$. Depth is initially $0$. Entering at  $s$ is considered a movement, so robots entering $s$ have depth $1$.  
\end{definition}

$d(A_i^{\mathcal{P}(n)}, t)$ is similarly defined with respect to $\mathcal{P}(n)$.

\begin{definition}
Let $T$ be a tree graph environment (such as $\mathcal{P}(n)$) with source vertex $s$. A vertex $v$ of $T$ becomes \textbf{slow} at time $t$ if a mobile robot on $v$ was activated and found no vertex it could move to, and also, either $v$ is a leaf of $T$ or all of its descendants in $T$ are slow at time $t$. 

A robot $A_i$ is \textbf{slow} at time $t$ if it is located at a slow vertex at time $t$.
\end{definition}

\begin{definition}
The \textbf{slow makespan} of $S$ on $T$, $M_{slow}^T$, is the first time all vertices of $T$ are slow when simulating the event order $S$.
\end{definition}

$G$ is not always a tree, but given a fixed event order $S$, we can associate to $S$ a spanning tree of $G$, $T_S$, containing $G(t)$ as a subtree for all times $t$. Lemma \ref{treelemma} says robots only use edges of  $T_S$, so we may define the slow makespan of $S$ on the $G$-simulation as the slow makespan on $T_S$. Slow makespan is clearly also defined for the $\mathcal{P}(n)$-simulation.  Furthermore, $M_{slow}^G$ is an upper bound on the (regular) makespan of the $G$-simulation, since every vertex must have a settled robot before it becomes slow and, as the settled robots of $G$ never move, they cannot be deleted by $S$.

Our motivation for introducing slow makespan is that we wish to show $\mathcal{P}(n)$ is the environment that maximizes slow makespan on $n$ vertices. However, it does not maximize normal makespan (see Table \ref{makespancomparetable} for an example). 

\begin{lemma}
A slow robot $A_i^G$ is forever unable to move and never deleted in the event order $S$.
\label{slowdeletelemma}
\end{lemma}

\begin{proof}
Only robots attempting to move can be deleted. If $A_i^G$ is at a leaf of $T_S$, it can never move, since its parent vertex in $T_S$ contains a settled robot marking the vertex of a robot in a different location, and settled robots are never deleted. Hence, $A_i^G$ is never deleted. Slow vertices propagate upwards from the leaves of $T_S$, so the statement of the lemma follows by induction.
\end{proof}

\begin{proposition}
For any event order $S$, $M_{slow}^G \leq M_{slow}^{\mathcal{P}(n)}$. \label{slowmakespancoupling}
\end{proposition}

An intuitive argument for this proposition is that if the spanning tree $T_S$ of $G$ is not $\mathcal{P}(n)$, then some vertex $v$ of $T_S$ must have multiple descendants, hence robots entering $v$ will be able to branch to different neighbours and $v$ is less likely to be blocked. Consequently, robots will enter $G$ faster than $\mathcal{P}(n)$, and so  $M_{slow}^G \leq M_{slow}^{\mathcal{P}(n)}$. We need to formalize this intuition into an argument that holds for any event order $S$. It turns out there are many subtleties involving asynchronicity, settling and crashing which make this not straightforward, and we require a rather technical argument. (Such subtleties are also why it is simpler to compare the environments $G, \mathcal{P}(n)$, $\mathcal{P}(\infty)$, $\mathcal{P^*}(\infty), B$ rather than compare $G$ to $B$ directly.)

We prove Proposition \ref{slowmakespancoupling} by induction on the \textit{meaningful event times} $t_0, t_1, \ldots$ in the event order $S$. We show the following statements to be true for non-deleted robots at all times $t_m$:

\begin{enumerate}[label=(\alph*)]
    \item If $A_i^G$ is not slow or settled, then $d(A_i^G, t_m) \geq d(A_i^{\mathcal{P}(n)}, t_m)$. 
    
    \item If $A_i^{\mathcal{P}(n)}$ is slow or settled, then $A_i^G(t_m)$ is slow or settled, and $d(A_i^G, t_m) \leq d(A_i^{\mathcal{P}(n)}, t_m)$.
\end{enumerate}

We note that both statements are (trivially) true at time $t_0$, as no event has occurred yet. 

\begin{lemma}
If statement (b) is true up to time $t_m$, settled and slow robots of $\mathcal{P}(n)$ neither move nor get deleted as a result of an event of $S$ scheduled for time $t_m$ (i.e., the robots still exist and are in the same place at time $t_{m+1}$).

\label{neverdeleteslowcorollary}
\end{lemma}

Assuming (a) and (b) hold at all times, let us see how to infer Proposition \ref{slowmakespancoupling}. If a vertex becomes slow at some time $t$, it must contain a settled and a mobile robot, both of whom become  slow. Lemma \ref{neverdeleteslowcorollary} says that slow and settled robots of $\mathcal{P}(n)$ never get deleted. Hence, the first time there are $2n$ slow robots on $\mathcal{P}(n)$ (two at every vertex) is $M_{slow}^{\mathcal{P}(n)}$. Statement (b) implies that if $\mathcal{P}(n)$ has $2n$ slow robots, $G$ must also contain $2n$ slow or settled robots. It is immediate to verify that this can only happen when $G$ has $2n$ slow robots. Hence, at time $M_{slow}^{\mathcal{P}(n)}$, $G$ has $2n$ slow robots--two at every vertex. The inequality  $M_{slow}^{\mathcal{P}(n)} \geq M_{slow}^G$ follows by definition.  \qed

\begin{lemma} 
If statements (a) and (b) are true up to time $t_{m}$, statement (a) is true at time $t_{m+1}$.
\label{statementaproof}
\end{lemma}

\begin{lemma}
If statements (a) and (b) are true up to time $t_m$, statement (b) is true at time $t_{m+1}$. 
\label{statementbproof}
\end{lemma}

\subsubsection{$\mathcal{P}(n)$ versus $\mathcal{P}(\infty)$}
\label{pathgraphsection}

We wish to bound $M_{slow}^{\mathcal{P}(n)}$ (which is determined by the event order $S$). We do this by comparing simulations of $S$ on different environments. To start, let $\mathcal{P}(\infty)$ be the path graph with infinite vertices, and where $s = v_1$. We may simulate $S$ on $\mathcal{P}(\infty)$ as we did on  $\mathcal{P}(n)$. 

\begin{lemma}
For any event order $S$ simulated on $\mathcal{P}(n)$ and $\mathcal{P}(\infty)$ and any time $t < M_{slow}^{\mathcal{P}(n)}$, $\mathcal{P}(n)$ and $\mathcal{P}(\infty)$ contain the exact same number of robots.
\label{Pinftylemma}
\end{lemma}

\begin{proof}
The configuration of robots in the first $n$ vertices of  $\mathcal{P}(n)$ and $\mathcal{P}(\infty)$ is identical until $v_n$ becomes slow in $\mathcal{P}(n)$. After $v_n$ becomes slow, the configuration of robots in the first $n-1$ vertices is still the same in both graphs until a robot in $v_{n-1}$ is prevented from moving by a robot in $v_n$, meaning $v_{n-1}$ becomes slow. By induction, the configuration of robots in the first $k$ vertices of both graphs is identical until $v_k$ in $\mathcal{P}(n)$ becomes slow (we use Lemma \ref{neverdeleteslowcorollary} to infer that the slow robots at $v_{k+1}$ are never deleted). Hence, until $v_1$ becomes slow, robots enter at the same times in $\mathcal{P}(n)$ and $\mathcal{P}(\infty)$. $v_1$ becomes slow precisely at time  $M_{slow}^{\mathcal{P}(n)}$.
\end{proof}

\subsubsection{$\mathcal{P}(\infty)$ versus $\mathcal{P}^*(\infty)$} 
We simulate $S$ on the environment $\mathcal{P}^*(\infty)$. $\mathcal{P}^*(\infty)$ is $\mathcal{P}(\infty)$ with the modification that there is at time $t=0$ a settled robot at every vertex $v_i$. The settled robot at $v_i$ marks $v_{i-1}$. These ``dummy'' robots are never activated, and are not of the indexed robots $A_1, A_2, \ldots$. Because there is already a settled robot at every vertex, the robots $A_1, A_2, \ldots$ never become settled. Call this environment $\mathcal{P}^*(\infty)$. Lemma \ref{Pstarinftyslowerlemma} shows $\mathcal{P}^*(\infty)$ is strictly slower than $\mathcal{P}(\infty)$:

\begin{lemma}
For any event order $S$ and at any time $t$, the amount of \textbf{mobile}-state robots in  $\mathcal{P}^*(\infty)$ at time $t$ is at most the total amount of robots in $\mathcal{P}(\infty)$.
\label{Pstarinftyslowerlemma}
\end{lemma}

\subsubsection{$\mathcal{P}^*(\infty)$ versus totally asymmetric simple exclusion}

We bound the arrival rate of robots at $\mathcal{P}^*(\infty)$ by another, even slower process. This process, $B$, takes place on the path graph $\mathcal{P}(\infty)$ where we also have non-positive vertices $v_0, v_{-1}, v_{-2}, \ldots$, and such that there is an edge $(v_i, v_{i+1})$ for every $i$. Like $\mathcal{P}^*(\infty)$ there is initially a settled robot at every vertex, marking the vertex before it. Unlike the other processes, robots do not enter at $s$: the robot $A_{i}$ begins inside the graph environment as a mobile robot located at $v_{-i+1}$. To compare $B$ with $\mathcal{P}^*(\infty)$, we count the robots that cross the edge $(v_0, v_1)$. There is one more crucial feature of $B$: robots are never deleted from $B$. Scheduled robot deletions at $S$ are treated as a regular activation of the robot. Besides these differences, $S$ can be simulated on $B$ as before.

\begin{lemma}
For any event order $S$ and at any time $t$, the number of mobile robots that crossed the $(v_0, v_1)$ edge of $B$ is at most the number of robots that entered or were deleted before entering $\mathcal{P}^*(\infty)$.
\label{TASEPboundPstarlemma}
\end{lemma}

Recall that $S$ is an event order of some execution of Algorithm \ref{alg:localrule} on the graph environment of interest, $G$. We may randomly sample $S$ by running Algorithm \ref{alg:localrule} on $G$ and logging the events.

The stochastic process resulting from simulating a \textit{randomly sampled} event order $S$ on $B$ is called a  \textit{totally asymmetric simple exclusion process} (TASEP) with step initial condition, first introduced in \cite{spitzer1991interaction}. In this process, robots (called also ``particles'') are activated at exponential rate $1$ and attempt to move rightward whenever no other robot blocks their path. This is precisely the outcome of simulating $S$ on $B$ (since robot activations that lead to a deletion in the other processes are treated as a regular activation in $B$).

In TASEP with step initial condition, let us write $B_t$ to denote the number of robots that have crossed $(v_0, v_1)$ at time $t$. It is shown in \cite{rost1981non} that $B_t$ converges to $\frac{1}{4}t$ asymptotically almost surely (i.e., with probability 1 as $t \to \infty$). \cite{Johansson2000} shows that the deviations are of order $t^{1/3}$. Specifically we have in the limit: 

\vspace{-1em}

\begin{equation} \label{tasepequation}
    \lim_{t\to\infty} \mathbb{P}(B_t - \frac{t}{4} \leq 2^{-4/3} st^{1/3}) = 1 - F_2(-s)
\end{equation}

Valid for all $s \in \mathbb{R}$, where $F_2$ is the Tracy-Widom distribution and obeys the asymptotics $F_2(-s) = O(e^{-c_1 s^3})$ and $1-F_2(s) = O(e^{-c_2 s^{3/2}})$ as $s \to \infty$. We employ Equation \ref{tasepequation} and the prior analysis to prove Theorem \ref{alg1performancethm}:

\begin{proof}
Let $G$ be a graph environment with $n$ vertices. Let $S$ be the randomly sampled event order of an execution of Algorithm \ref{alg:localrule} on $G$. We will bound the slow makespan, $M_{slow}^G$.

We simulate $S$ over the environments $\mathcal{P}(n)$, $\mathcal{P}(\infty)$,  $\mathcal{P}^*(\infty)$, and $B$. From Lemma  \ref{TASEPboundPstarlemma} we know that at all times the number of robots that crossed the $(v_0, v_1)$ edge of $B$,  meaning $B_t$, is less than the number of robots that entered $\mathcal{P}^*(\infty)$ or were deleted before entering. At most $ct/4$ robots are deleted by time $t$, so the number of mobile robots at $\mathcal{P}^*(\infty)$ at time $t$ is at least $B_t - ct/4$. Lemmas \ref{Pinftylemma} and \ref{Pstarinftyslowerlemma} imply this is at least the number of robots at $\mathcal{P}(n)$ at any time $t < M_{slow}^{\mathcal{P}(n)}$.

At any time $t$, there cannot be more than $2n$ robots at $\mathcal{P}(n)$. Hence, if $B_t - ct/4 > 2n$, then $t \geq  M_{slow}^{\mathcal{P}(n)}$. By Proposition \ref{slowmakespancoupling}, we shall then also have $t \geq  M_{slow}^{G}$.

Write $t_n = 8((1-c)^{-1}+n^{-1/3})n$. We wish to show $t_n$ is an upper bound on $M_{slow}^G$ asymptotically almost surely, which is precisely the statement of Theorem \ref{alg1performancethm}. To show this, we are interested in  $X_n = \mathbb{P}(B_{t_n} \leq 2n)$, the probability that $B_{t_n}$ is less than $2n$ at time $t_n$. Showing $X_n$ tends to $0$ as $n \to \infty$ completes our proof. Define the probability

\vspace{-0.75em}
\begin{equation} 
    p(n,s) =
    \mathbb{P}(B_{t_n} - \frac{t_n}{4} \leq 2^{-4/3} s t_n^{1/3})
\end{equation}
\vspace{-1.25em}

$p(n,s)$ is the parametrized left innermost part of Equation \ref{tasepequation} with $t = t_n$ ($n$ is a positive integer). Note that $p(n,s)$ is monotonic increasing in $s$. Define $s_n = (2n-t_n/4)2^{4/3} t_n^{-1/3}$. By algebra, we have $X_n = p(n, s_n)$. Fix any constant $\mathcal{S}^*$ and define $Y_n = p(n, \mathcal{S}^*)$. Again by algebra, $s_n$ tends to $-\infty$ as $n \to \infty$. Hence, for a large $n$, we must have $s_n < \mathcal{S}^*$ and therefore $X_n \leq Y_n$ (by the monotonicity of $p(n,s)$). By Equation \ref{tasepequation}, $Y_n$ tends to $1 - F_2(-\mathcal{S}^*)$ as $n \to \infty$. Hence $X_n$ is at most $1 - F_2(-\mathcal{S}^*)$ in the limit. By taking $\mathcal{S}^* \to -\infty$ we see that $X_n$ in the limit is at most $1 - F_2(\infty) = 0$. \end{proof}

We note that slow makespan can be nearly equal to makespan (see Table \ref{makespancomparetable}, or consider a path graph $\mathcal{P}(n)$ the source vertex placed at $s = v_2$ and robots initially moving rightwards). Hence, one does not ``miss out'' on much by using it to bound makespan.

\subsection{Synchronous time and multiple sources}
We describe extensions of our results to two settings.

\textbf{Synchronous time}. We may consider a synchronous time setting that is discretized to steps $t = 1, 2, \ldots$ such that at every step, all the robots activate at once. In this setting, Algorithm \ref{alg:localrule} ends up exploring just one branch of the tree at a time, like depth-first-search; so no two robots ever attempt to enter the same vertex. Analysis similar to the asynchronous case shows that robots then enter at rate $t/2$ (instead of approximately $t/4$) on $\mathcal{P}(n)$, and  analogous reasoning to Lemma \ref{slowmakespancoupling} and Theorem \ref{alg1performancethm} gives an upper bound of $4(1-c)^{-1} n$ on the makespan of a graph with $n$ vertices, assuming $ct/2$ adversarial events. Consider the path graph $\mathcal{P}(n)$ with $s = v_2$ (not the usual $s = v_1$), and where the robots first fill the vertices $v_3, v_4, \ldots$ with a double layer before reaching $v_1$. The synchronous makespan of this environment is asymptotically $4(1-c)^{-1} n$. Hence, the  bound on the makespan in the synchronous case is exact.

\textbf{Multiple source vertices.} Instead of just having a single source vertex $s$, we may consider environments with multiple source vertices such that each of them corresponds to its own set of robots $A_1, A_2, \ldots$ entering over time. In asynchronous time, Lemma \ref{treelemma} can be generalized to show that $G(t)$ is then a \textit{forest}, and the robots attempt to create a spanning forest of $G$. The technique in this paper can be generalized to show that the makespan bound of Theorem \ref{alg1performancethm} holds. In general graph environments multiple sources may not improve the makespan by much. For example, consider the path graph $\mathcal{P}(n)$ with $k$
 sources on $v_1, v_2, \ldots v_k$. The makespan of this graph is bounded below by the makespan of the path graph  $\mathcal{P}(n-k-1)$ with a single source vertex $v_1$. 
 
\section{Simulation and evaluation}
\label{simulationsection}

For empirical confirmation of our analysis, we numerically simulated our algorithm on a number of environments. On these environments, we measured the makespan and the percentage of robots that crashed for the parameters $c = 0, 0.25, 0.75$, averaging them over 30 simulations per configuration and rounding to the nearest integer. Data on several environments is found in Table \ref{makespancomparetable}. Figure \ref{fig:simulation2} shows stills from some simulations. 

\vspace{-1em}

\begin{figure}[!htb]
    \centering 
    
\medskip
\begin{subfigure}{0.145\textwidth}
  \includegraphics[width=\linewidth]{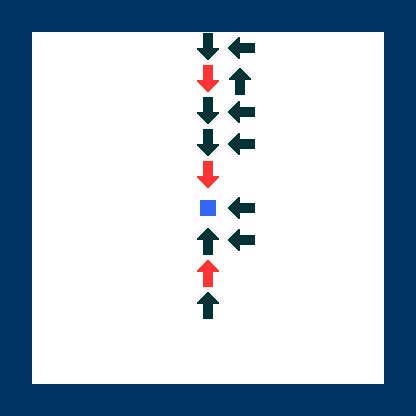}
\end{subfigure}\hfil 
\begin{subfigure}{0.145\textwidth}
  \includegraphics[width=\linewidth]{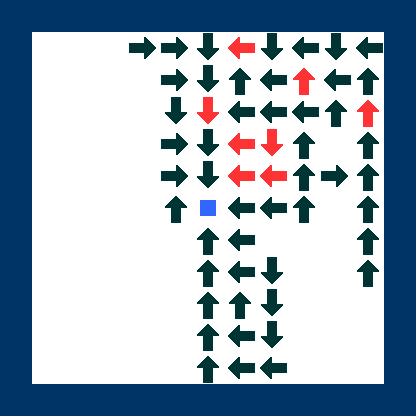}
\end{subfigure}\hfil 
\begin{subfigure}{0.145\textwidth}
  \includegraphics[width=\linewidth]{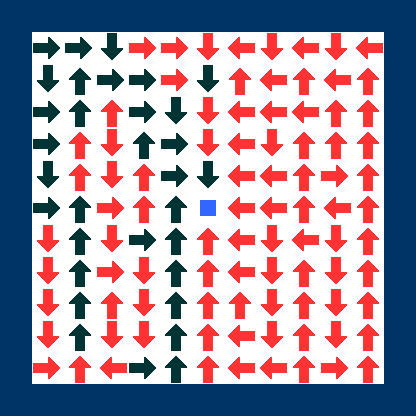}
\end{subfigure}

\vspace{-0.25em}
\medskip 
\begin{subfigure}{0.145\textwidth}
  \includegraphics[width=\linewidth]{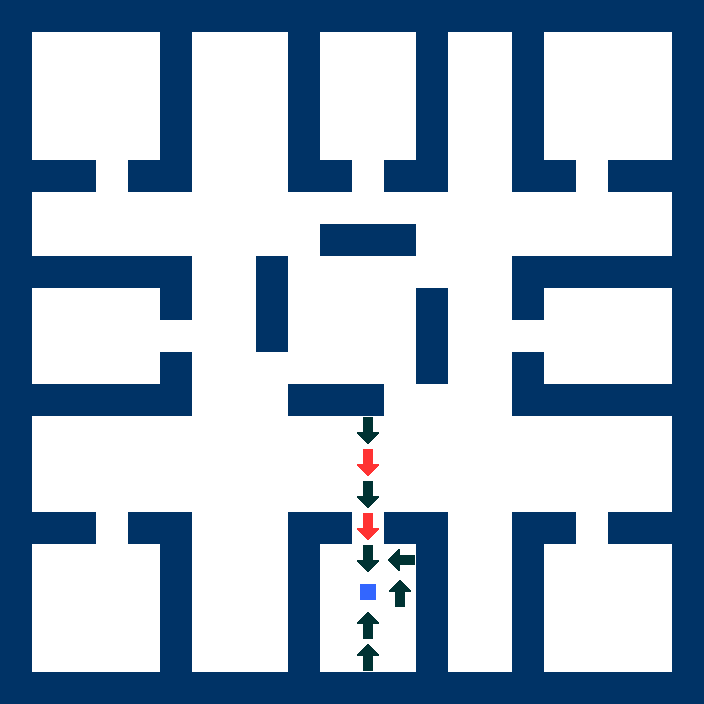}
\end{subfigure}\hfil 
\begin{subfigure}{0.145\textwidth}
  \includegraphics[width=\linewidth]{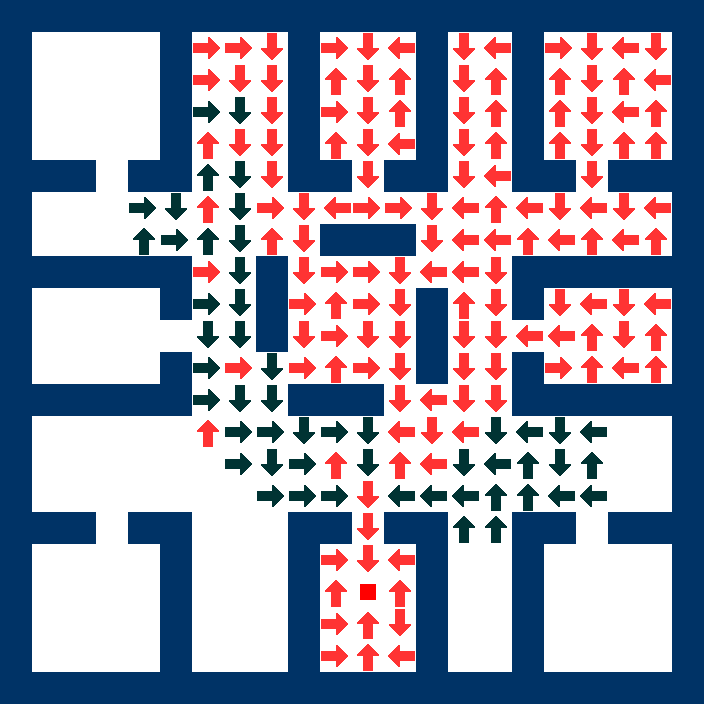}
\end{subfigure}\hfil 
\begin{subfigure}{0.145\textwidth}
  \includegraphics[width=\linewidth]{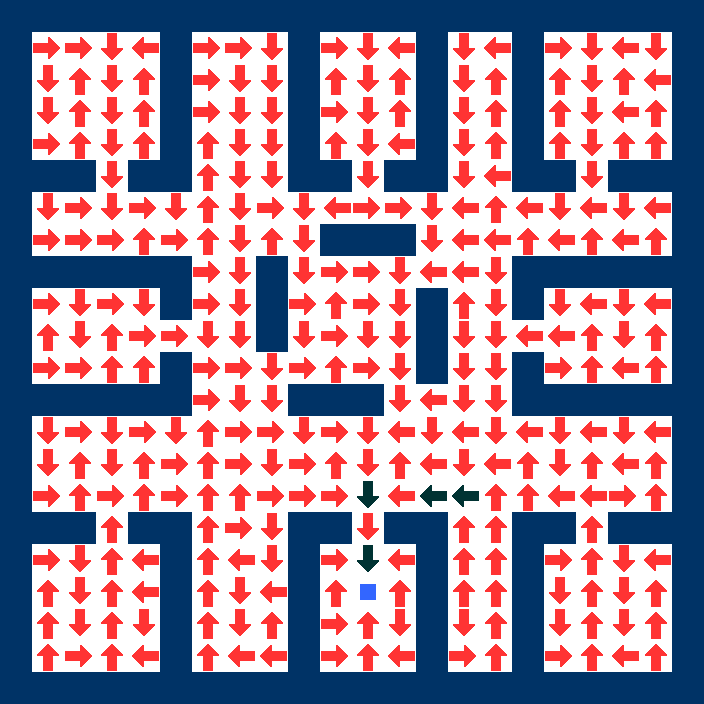}
\end{subfigure}

\caption{An execution of Algorithm \ref{alg:localrule} on (a) an $11\times11$ square grid and (b) an ``indoor'' environment. The legend is same as Figure \ref{fig:simulation}.}
\label{fig:simulation2}
\vspace{-1.5em}
\end{figure}

\begin{table}[!htb]
\begin{tabular}{lllll}
 & n &  $c=0$ & $c=0.25$ & $c=0.75$ \\ \hline
Figure \ref{fig:simulation} &  62  &      272;373  &  321;446 (18\%)   &      555;743 (52\%)     \\ 
Figure \ref{fig:simulation2}, (a)   &  121 & 463;554    &  484;586 (13\%)  &  715;921 (39\%)  \\  
Figure \ref{fig:simulation2}, (b)    &    300  & 1791;1907   &  2154;2281 (19\%)   &  3894;4116 (56\%)    \\ 
$\mathcal{P}(300)$    & 300  & 1677;2325   &  1940;2883  (23\%) &   3727;6147 (66\%)   \\ 
\end{tabular}
\caption{The makespan and slow makespan of Algorithm \ref{alg:localrule} over the environments in the referred-to Figures and over the path graph of length $300$, $\mathcal{P}(300)$ . We vary the crash density parameter $c$. The cell format is \textit{makespan ; slow makespan (\% of robots crashed)}. The column ``$n$'' gives the number of vertices in the environment.}
\label{makespancomparetable}
\vspace{-2em}
\end{table}

From the data, it is clear that makespan is affected by the shape of the environment and by $c$. We see that an increase in the percentage of robots crashed scales makespan up gracefully, and that spacious environments generally have lower makespans. We also confirm that the slow makespans are always lower than the bound of Theorem \ref{alg1performancethm}. Closest to the bound is the scenario where the environment is the path graph  $\mathcal{P}(300)$ and $c=0$, in which case slow makespan is almost exactly the bound, $8 \cdot 300$. This is consistent with our analysis that the $\mathcal{P}(n)$ environment has the largest slow makespan. It also verifies that Theorem \ref{alg1performancethm} gives a correct upper bound. Such data further suggests that for spacious environments, and for large $c$, performance on average is better than the \textit{worst-case} performance guarantee of Theorem \ref{alg1performancethm}. In the simulations, we did not choose our adversarial events to be maximally obstructing, but rather crashed robots arbitrarily--a cleverer adversary would cause the makespan and slow makespan to be closer to the worst-case (and cause a larger percentage of robots to crash).

\section{Discussion}
In swarm robotics, where one must coordinate an enormous robotic fleet, we must anticipate many faults, such as crashing and traffic jams. Because robots in the swarm are usually assumed to be autonomous and have limited computational power, complex techniques for handling such faults are not necessarily feasible. Hence, it is important to ask whether simple rules of behaviour can be effective. To this end, we investigated the problem of covering an unknown graph environment, and constructing an implicit spanning tree, with a swarm of frequently crashing robots. We showed a simple and local rule of behaviour that enables the swarm to quickly and reliably finish this task in the presence of crashes. The swarm's performance degrades gracefully as crash density increases. 

We outline here several directions for future research. First, our model interprets the ``swarm'' part of swarm robotics as a vast and redundant fleet of robots that can be dispersed into the environment over time. We used this model for uniform dispersal, but it would be interesting to adapt it to other kinds of missions, and design algorithms for those missions that can handle crashes or other forms of interference. For example, in \cite{arxivprobabilisticpursuits}, mobile agents entering at a source node $s$ over time sequentially pursue each other to discover shortest paths between $s$ and some target node $t$. The algorithm succeeds even if some of the agents are interrupted and have their location changed.

Next, in this work, we made the simplifying assumption that the environment of the robots is discrete. If the robots instead attempted to cover a continuous planar domain by an algorithm similar to ours, the robots would need to construct a shared discrete graph representation of the environment through the settled robots in $G(t)$ and their markings. We believe that our algorithm can readily be extended to such settings.

Lastly, can we exploit the large number of robots in a swarm to handle other kinds of errors? There are many situations and modes of failure that can be discussed, such as Byzantine robotic agents, or dynamic changes to the environment.


\bibliographystyle{plainnat}
\bibliography{references}

\newpage
\section{Fast Uniform Dispersion of a Swarm - Supplementary Appendix}

Reminder:

\begin{enumerate}[label=(\alph*)]
    \item If $A_i^G$ is not slow or settled at time $t_m$, then $d(A_i^G, t_m) \geq d(A_i^{\mathcal{P}(n)}, t_m)$. 
    
    \item If $A_i^{\mathcal{P}(n)}$ is slow or settled at time $t_m$, then $A_i^G(t_m)$ is slow or settled, and $d(A_i^G, t_m) \leq d(A_i^{\mathcal{P}(n)}, t_m)$.
\end{enumerate}

\subsection{Proof of Lemma \ref{neverdeleteslowcorollary}}

\begin{proof}
Referring to the Lemma's statement, we remind that here  ``time $t_m$'' refers to the configuration of agents at time $t_m$ \textit{before} any scheduled events. Hence, even if something is true at time $t_m$, we still need to show that it remains true after the events that happen at time $t_m$. 

Let $A_i^{\mathcal{P}(n)}$ be slow or settled at time $t_m$. To show $A_i^{\mathcal{P}(n)}$ will not be deleted, it suffices to show the event order $S$ will not delete $A_i^G$. (b) implies  $A_i^G$ is settled or slow at time $t_m$. Lemma \ref{slowdeletelemma} says $S$ never deletes slow agents. $S$ never deletes settled agents of $G$ as, in our model, agents are only deleted when they move, and $S$ obeys the rules of the model when simulated on $G$. Hence, $S$ will not delete $A_i^G$. 

Next we show that $A_i^{\mathcal{P}(n)}$ will not move as a result of an event scheduled for time $t_m$. If $A_i^{\mathcal{P}(n)}$ is settled, this is true by definition. Otherwise, $A_i^{\mathcal{P}(n)}$ is slow. By assumption, (b) is true at \textit{all} times up to $t_m$. Hence, by the same reasoning as the above paragraph, agents of $\mathcal{P}(n)$ that became slow or settled at or prior to time $t_m$ have not been deleted. Consequently, the argument of Lemma \ref{slowdeletelemma} applies also here, allowing us to conclude that agents cannot move after they become slow. In particular this applies to $A_i^{\mathcal{P}(n)}$.
\end{proof}

\subsection{Proof of Lemma \ref{statementaproof}}

\begin{proof}
Only one event occurs at time $t_m$. This event is either an uninterrupted activation of an agent (meaning the agent is not deleted), or an activation that leads to a deletion. If the event is a deletion, (a) holds at time $t_{m+1}$ trivially, so we assume that it is an uninterrupted activation.

Let $A_i^G$ and $A_i^{\mathcal{P}(n)}$ be the agents that are activated at time $t_m$. The depth of any other agent is unchanged, so we need only verify (a) for these two agents. Assuming (a) it true at time $t_m$, it is only possible for (a) to become false at time $t_{m+1}$ if $A_i^G$ did not move, but $A_i^{\mathcal{P}(n)}$ did. We assume this is the case. 

If $A_i^G$ does not move as a result of its activation at time $t_m$, then either it is settled, in which case (a) is true and we are done, or there is a mobile agent at every neighbouring vertex in $G(t_m)$. If $A_i^G$ is mobile and all of its neighbours are slow at time $t_m$, then $A_i^G$ becomes slow at time $t_{m+1}$ and (a) is true. Otherwise there is a mobile agent, $A_j^G$, that is preventing $A_i^G$ from moving and is not slow. We must have that

\begin{equation}
    d(A_i^G, t_{m+1}) + 1
    \label{eq0lemma} = d(A_j^G, t_{m+1})
\end{equation}

Because $A_i^G$ and $A_j^G$ are always moving down a spanning tree $T_S$ of $G$, hence the depth of $A_j^G$ must be precisely one greater than $A_i^G$'s in order to prevent movement. 

Because $A_j^G$ is not activated at time $t_m$, (a) and (b) are still true for it at time $t_{m+1}$. Because $A_j^G$ is not slow or settled, (a) implies that 

\begin{equation}
    d(A_j^G, t_{m+1}) \geq d(A_j^{\mathcal{P}(n)}, t_{m+1})
    \label{eq1lemma}
\end{equation}

And the contrapositive of (b) implies that $A_j^{\mathcal{P}(n)}$ is not settled. However, consider the structure of the graph $\mathcal{P}(n)$: if $A_j^{\mathcal{P}(n)}$ is mobile, then since it entered before $A_i^{\mathcal{P}(n)}$, it must be further ahead. In particular, we must have 

\begin{equation}
    d(A_j^{\mathcal{P}(n)}, t_{m+1}) \geq d(A_i^{\mathcal{P}(n)}, t_{m+1}) + 1
    \label{eq2lemma}
\end{equation} 

As otherwise $A_j^{\mathcal{P}(n)}$ would have prevented $A_i^{\mathcal{P}(n)}$ from moving when activated at time  $t_m$.

(In)equalities \ref{eq0lemma}, \ref{eq1lemma} and \ref{eq2lemma} imply $d(A_i^G, t_{m+1}) \geq  d(A_i^{\mathcal{P}(n)}, t_{m+1})$. This shows (a) is true at time $t_{m+1}$.
\end{proof}

\subsection{Proof of Lemma \ref{statementbproof}}

\begin{proof}
As in Lemma \ref{statementaproof}, we can assume that the event at time $t_m$ is the uninterrupted activation of a pair of agents $A_i^G$ and $A_i^{\mathcal{P}(n)}$, and we need only verify that (b) is still true for this pair of agents. We separate our proof into cases.

\textbf{Case 1:} Assume $A_i^{\mathcal{P}(n)}$ is settled at time $t_{m+1}$. Because $\mathcal{P}(n)$ is a path graph and using Lemma \ref{neverdeleteslowcorollary}, $A_i^{\mathcal{P}(n)}$ can only be settled if every non-deleted agent that entered before it is settled behind it. At time $t_{m+1}$ (b) is still true for all agents other than $A_i^G$ and $A_i^{\mathcal{P}(n)}$. Hence, it follows from (b) that for any non-deleted agent $A_j^G$ where $j < i$ we have:

\begin{equation} 
\label{equationbstatement1}
d(A_j^G, t_{m+1}) \leq d(A_j^{\mathcal{P}(n)}, t_{m+1})
\end{equation}

Algorithm \ref{alg:localrule} guarantees that any agent in $G$ always neighbours a settled agent or is at the same location as a settled agent. Thus, we know that $d(A_i^G, t_{m+1}) \leq d(A_j^G, t_{m+1}) + 1$ for some settled agent $A_j$. Furthermore, this inequality must hold for some $A_j^G$ that entered \textit{before} $A_i^G$ (i.e., $j < i$), because any settled agent that entered after $A_i^G$ must have gone down a different branch of $T_S$, otherwise it would be blocked by $A_i^G$ and unable to settle. Let $j_{max} = max_{j < i} d(A_j^G, t_{m+1})$. Then $d(A_i^G, t_{m+1}) \leq j_{max} + 1$. If this is an equality, $A_i^G$ is necessarily settled.

From Inequality \ref{equationbstatement1} we infer 

\begin{equation}
    d(A_i^{\mathcal{P}(n)}, t_{m+1}) \geq max_{j < i} d(A_j^{\mathcal{P}(n)}, t_{m+1})  + 1 \geq j_{max} + 1 \geq d(A_i^G, t_{m+1})
\end{equation}

Where $d(A_i^{\mathcal{P}(n)}, t_{m+1}) \geq max_{j < i} d(A_j^{\mathcal{P}(n)}, t_{m+1})  + 1$ follows from the fact that $A_i^{\mathcal{P}(n)}$ is ahead of all non-deleted agents that came before it. In the case of equality, $A_i^G$ must be settled. If $A_i^G$ isn't settled, then the inequality is strict. Consequently, it follows from the fact that (a) holds at time $t_{m+1}$ (Lemma \ref{statementaproof}) that $A_i^G$ must be slow. Otherwise,  (a) implies $A_i^G$'s depth is greater than $A_i^{\mathcal{P}(n)}$'s, contradicting the inequality. Either way, (b) is true.

\textbf{Case 2:} Assume $A_i^{\mathcal{P}(n)}$ is slow and not settled at time $t_{m+1}$. If $A_i^{\mathcal{P}(n)}$ is slow at $t_m$, then it follows from (b) that $A_i^G$ is slow or settled at $t_m$, and so activation cannot affect either of these agents, meaning (b) remains true at $t_{m+1}$ and we are done. Thus, we may assume $A_i^{\mathcal{P}(n)}$ is not slow at time $t_m$. 

Using Lemma \ref{neverdeleteslowcorollary}, $A_i^{\mathcal{P}(n)}$ can only become slow at time $t_{m+1}$ if all vertices behind it contain settled agents, and all vertices ahead of it contain two slow agents (one settled and one mobile). If  $d(A_i^{\mathcal{P}(n)}, t_{m})$ is $k$ there are $n+(n-k)=2n-k$ slow or settled agents in $\mathcal{P}(n)$ at time $t_m$. These $2n-k$ agents must have entered $\mathcal{P}(n)$ before $A_i^{\mathcal{P}(n)}$, because any agent that enters after $A_i^{\mathcal{P}(n)}$ must pass it to become slow or settled, and this is impossible because $A_i^{\mathcal{P}(n)}$ is not settled. 

Using (b) we learn from the above that in $G$, at time $t_m$ there are at least  $2n-k$ settled or slow agents that entered before $A_i^G$. Of these, at least $n-k$ agents are slow and mobile, and have greater depth than $A_i^G$ or are in a different branch of $T_S$ (because they arrived before $A_i^G$ and $A_i^G$ could not have passed them). There are thus at most $n - (n-k) = k$ vertices $A_i^G$ could have visited since entering $G$, meaning its depth is at most $k$, and we have $d(A_i^G, t_{m}) \leq d(A_i^{\mathcal{P}(n)}, t_{m})$. 

If this inequality is strict, then from statement (a) we learn that $A_i^G$ is settled or slow, so (b) is true and we are done. Otherwise, $d(A_i^G, t_m) = k$. We saw there are (at least) $n-k$ slow mobile agents in $G$ that have greater depth than $A_i^G$ or are in a different branch of $T_S$. From this, we infer that any descendant of $A_i^G$ must contain a slow mobile agent, or that $A_i^G$ is at a leaf of $T_S$ and has no descendants. Thus, if $A_i^G$ is not already settled or slow, it will become slow after the activation at time $t_m$, since its slow descendants will prevent it from moving. This completes the proof.

\end{proof}

\subsection{Proof of Lemma \ref{Pstarinftyslowerlemma}}

\begin{proof}
Let $A_i^*$ be the copy of $A_i$ simulated over $\mathcal{P}^*(\infty)$. Let $t_0, t_1, t_2, \ldots $ be the meaningful event times of $S$. We show by induction that at any time $t_m$, for all non-deleted agents: 

\begin{equation}
     \textrm{either $A_i^{\mathcal{P}(\infty)}$ is settled  or } d(A_i^*, t_m) \leq d(A_i^{\mathcal{P}(\infty)}, t_m)
    \label{lemmapstarequation}
\end{equation}

This implies any agent that enters $\mathcal{P}^*(\infty)$ must have already or concurrently entered $\mathcal{P}(\infty)$, completing the proof.

The induction statement is trivially true at time $t_0$, as no event has occurred yet. We assume it is true up to time $t_m$, and show it remains true at $t_{m+1}$. 

If the event scheduled for time $t_m$ was a deletion of some agent, the statement remains trivially true (as both simulated versions of the agent are deleted). Otherwise, the scheduled event is the uninterrupted activation of some pair of agents $A_i^*$ and $A_i^{\mathcal{P}(n)}$. 

Any agent $A_j$ where $j \neq i$ does not move, so we need only verify the inductive statement remains true for $A_i^*$ and $A_i^{\mathcal{P}(n)}$. The only situation in which Inequality \ref{lemmapstarequation} is falsified at time $t_{m+1}$ if it is true at time $t_m$ is if $d(A_i^*, t_m) = d(A_i^{\mathcal{P}(\infty)}, t_m)$ and $A_i^{\mathcal{P}(\infty)}$ is mobile at time $t_{m+1}$, but $A_i^*$ manages to move whereas $A_i^{\mathcal{P}(\infty)}$ is blocked by a mobile agent $A_j^{\mathcal{P}(\infty)}$. By the inductive hypothesis, $d(A_j^*, t_m) \leq d(A_j^{\mathcal{P}(\infty)}, t_m)$. Because $\mathcal{P}^*(\infty)$ is a path graph and $j < i$, we know that $d(A_j^*, t) >  d(A_i^*, t)$ at all times $t$ after $A_j^*$ entered the environment. Hence, if $d(A_i^*, t_m) = d(A_i^{\mathcal{P}(\infty)}, t_m)$ and $A_j^{\mathcal{P}(\infty)}$ blocks $A_i^{\mathcal{P}(\infty)}$, then $A_j^*$ must also block $A_i^*$ when it attempts to move. This shows that the inductive hypothesis is correct at time $t_{m+1}$.

\end{proof}

\subsection{Proof of Lemma  \ref{TASEPboundPstarlemma}}

\begin{proof} Unlike Lemma \ref{Pstarinftyslowerlemma}, here we count the number of agents that enter $\mathcal{P}^*(\infty)$, and not the number of currently existing agents that entered it. This means we count also agents that entered at $\mathcal{P}^*(\infty)$ but were deleted. This difference is necessary for the comparison, because agents cannot be deleted from $B$. 

Despite this difference, the proof is very similar to Lemma \ref{Pstarinftyslowerlemma}. One shows by induction on the meaningful event times $t_0, t_1, t_2, \ldots$  that at any time $t_m$, for any $i$ such that $A_i^{\mathcal{P}^*(\infty)}$ was not deleted we have:

\begin{equation}
    \label{Bprocessinequality}
    d(A_i^{B}, t_m)  - i + 1 \leq  d(A_i^{\mathcal{P}^*(\infty)}, t_m) 
\end{equation}

Note that $d(A_i^{B}, t_m)  - i + 1$ is the index of the vertex of $A_i^B$ at time $t_m$. If $A_i^{B}$ crossed $(v_0, v_1)$ we must have $d(A_i^{B}, t_m) - i + 1 \geq 1$. Recalling that if $A_i^{\mathcal{P}^*(\infty)}$ is outside of $G$ at time $t$ then $d(A_i^{\mathcal{P}^*(\infty)}, t) = 0$, we see by  (\ref{Bprocessinequality}) that crossing $(v_0, v_1)$ can only happen if  $A_i^{\mathcal{P}^*(\infty)}$ entered $\mathcal{P}^*(\infty)$, or if $A_i^{\mathcal{P}^*(\infty)}$    was deleted before entering. Hence, the Lemma follows from (\ref{Bprocessinequality}). 

Let us show (\ref{Bprocessinequality}) holds by induction. It holds trivially for all $i$ at $t_0$. Now, assume (\ref{Bprocessinequality}) holds at time $t_m$, and we will show  it holds at time $t_{m+1}$.

Suppose the pair of agents activated at $t_m$ is $A_i^{\mathcal{P}^*(\infty)}$ and $A_i^B$. Then these are the only agents for which (\ref{Bprocessinequality}) might be false at $t_{m+1}$. Assuming (\ref{Bprocessinequality}) is true at $t_m$, it can only become false at $t_{m+1}$ if $d(A_i^{B}, t_m) - i + 1 =   d(A_i^{\mathcal{P}^*(\infty)}, t_m)$, but $A_i^{B}$ successfully moves as a result of activation at time $t_m$ whereas $A_i^{\mathcal{P}^*(\infty)}$ does not and also is not deleted. If $A_i^{\mathcal{P}^*(\infty)}$ does not move this means some $A_j^{\mathcal{P}^*(\infty)}$, $j < i$ is blocking it. Hence, we must have $d(A_j^{\mathcal{P}^*(\infty)}, t_{m+1}) = d(A_i^{\mathcal{P}^*(\infty)}, t_{m+1})  + 1$. By the inductive hypothesis we have $d(A_j^B,t_{m+1}) - j + 1 \leq d(A_j^{\mathcal{P}^*(\infty)}, t_{m+1})$. Since $j < i$, $A_j^B$ is always ahead of $A_i^B$, meaning $d(A_i^B,t_{m+1}) - i + 1 < d(A_j^B,t_{m+1}) - j + 1$. Combining these (in)equalities we get  $d(A_i^B,t_{m+1}) - i + 1 < d(A_i^{\mathcal{P}^*(\infty)}, t_{m+1})  + 1$, hence $ d(A_i^{B}, t_{m+1})  - i + 1 \leq d(A_i^{\mathcal{P}^*(\infty)}, t_{m+1})$. This completes the proof by induction of (\ref{Bprocessinequality}).

\end{proof} 

\end{document}